\documentclass{article}
\usepackage{ijcai19}

\usepackage{theorem,latexsym,graphicx,amssymb}
\usepackage{complexity}
\usepackage{times}
\usepackage{soul}
\usepackage{url}
\usepackage[hidelinks]{hyperref}
\usepackage[utf8]{inputenc}
\usepackage[small]{caption}
\usepackage{graphicx}
\usepackage{amsmath}
\usepackage{booktabs}
\usepackage{algorithm}
\usepackage{algorithmic}
\usepackage{natbib}
\usepackage{color}
\newenvironment{proof}{{\bf Proof:  }}{\hfill\rule{2mm}{2mm}\vspace*{5pt}}

\newtheorem{theorem}{Theorem}[section]
\newtheorem{corollary}{Corollary}[section]

\newtheorem{definition}[theorem]{Definition}
\newtheorem{lemma}[theorem]{Lemma}

\newtheorem{example}[theorem]{Example}

\newcommand{\floor}[1]{\lfloor #1 \rfloor }
\newcommand{\tabincell}[2]{\begin{tabular}{@{}#1@{}}#2\end{tabular}}

\newcommand{\MMS}{\mathsf{MMS}}
\renewcommand{\SP}{\mathsf{SequPick}}
\newcommand{\RD}{\mathsf{RandDecl}}
\newcommand{\RR}{\mathsf{RounRobi}}

\title{Strategyproof and Approximately Maxmin Fair Share Allocation of Chores\thanks{The authors are ordered alphabetically. 
This work is partially supported by NSF CAREER Award No. 1553385.} }

\author{
Haris Aziz$^1$\and
Bo Li$^2$\And
Xiaowei Wu$^3$\\
\affiliations
$^1$UNSW Sydney and Data61 CSIRO, Australia\\
$^2$Department of Computer Science, Stony Brook University, USA\\
$^3$Faculty of Computer Science, University of Vienna, Austria\\
\emails
haziz@cse.unsw.edu.au,
boli2@cs.stonybrook.edu,
xiaowei.wu@univie.ac.at
}

\begin{document}

\maketitle

\begin{abstract}
	We initiate the work on fair and strategyproof allocation of indivisible chores.
	The fairness concept we consider in this paper is maxmin share (MMS) fairness. 
	We consider three previously studied models of information elicited from the agents: the ordinal model, the cardinal model, and the public ranking model in which the ordinal preferences are publicly known.
	We present both positive and negative results on the level of MMS approximation that can be guaranteed if we require the algorithm to be strategyproof. Our results uncover some interesting contrasts between the approximation ratios achieved for chores versus goods. 
\end{abstract}

\section{Introduction}

Multi-agent resource allocation is one of the major research topics in artificial intelligence~\citep{BCM15a}. We consider fair allocation algorithms of indivisible chores when agents have additive utilities. The fairness concept we use in this paper is the intensively studied and well-established maxmin share fairness. The maxmin fair share of an agent is the best she can guarantee for herself if she is allowed to partition the items but then receives the least preferred bundle.

In this paper we take a mechanism design perspective to the problem of fair allocation. We impose the constraint that the algorithm should be strategyproof, i.e., no agent should have an incentive for report untruthfully for profile of valuations.
The research question we explore is the following one. 
\emph{When allocating indivisible chores, what approximation guarantee of maxmin share fairness can be achieved by strategyproof algorithms?}
This approach falls under the umbrella of approximation mechanism design without money that has been popularized by \citet{PrTe13a}.

Maxmin share fairness  was proposed by \citet{Budi11a} as a fairness concept for allocation of indivisible items. The concept coincides with the standard proportionality fairness concept if the items to be allocated are divisible. There have been several works on algorithms that find an approximate MMS allocation~\citep{PrWa14a,AMNS15a,BM17a,SGHSY18,ARSW17a}. None of these works took a mechanism design perspective to the problem of computing approximately MMS allocation. \citet{ABM16a} were the first to embark on a study of strategyproof and approximately MMS fair algorithms. Their work only focussed on the case of goods. There are many settings in which agents may have negative utilities such as when chores or tasks are allocated. In this paper, we revisit strategyproof fair allocation by considering the case of chores.

\subsection{Our results}

\begin{table*}[h!]
	\begin{center}
		\begin{tabular}{ |l|ll|ll|}
			\hline
			& Goods & ~ & Chores & ~\\ \hline
			& Lower & Upper  & Lower  & Upper \\ \hline
			Ordinal & $\Omega(\log n)$& $O(m-n)$ &{\tabincell{c}{$\frac{4}{3}$ (D) \\$\frac{6}{5}$ (R)}} & {\tabincell{c}{$O(\log \frac{m}{n})$ (D) \\ $O(\sqrt{\log n})$ (R)}}  \\ \hline
			Cardinal &$2$ & $O(m-n)$ & {\tabincell{c}{$\frac{4}{3}$ (D)\\ N/A (R)}} & {\tabincell{c}{$O(\log \frac{m}{n})$ (D) \\ $O(\sqrt{\log n})$ (R)}} \\ \hline
			Public ranking& $\frac{6}{5}$ & {\tabincell{c}{$\frac{3}{2}$ for $n=2$ \\ $2$ for $n=3$ \\ $O(n)$ for any $n$}} & $\frac{6}{5}$ (D) & {\tabincell{c}{$\frac{3}{2}$ for $n\leq 3$ \\ $2$ for any $n$}} \\\hline
		\end{tabular}
	\end{center}
	\caption{
		Lower and upper bounds on approximation of MMS fairness of strategyproof algorithms. 
		The results for goods were proved by~\citet{ABM16a}, all of which concern deterministic algorithms.
		The results for chores are proved in this paper, where we use (D) and (R) to indicate deterministic and randomized algorithms, respectively. 
	}
	\label{table:lit}
\end{table*}

We initiate the study of maxmin share (MMS) allocations of $m$ indivisible chores among $n$ strategic agents.
It is assumed that all agents have underlying cardinal additive utilities over the chores. 
We consider three preference models in this work. 

\begin{itemize}
	\item \textbf{Cardinal model}: agents are asked to report their cardinal utilities over the items.
	\item \textbf{Ordinal model}: agents are only allowed or asked to express their ordinal rankings over the items.
	\item \textbf{Public ranking model}: all agents' rankings are public information and the agents are asked to report their utilities that are consistent the known ordinal rankings.
\end{itemize}

First, for cardinal and ordinal models, we design a deterministic sequential picking algorithm $\SP$, 
which is strategyproof and unexpectedly achieves an approximation of $O(\log {m\over n})$\footnote{In this paper we use $\log()$ to denote $\log_2()$.}.
Roughly speaking, given an order of the agents, a sequential picking algorithm lets each agent $i$ pick $a_{i}$ items and leave.
\citet{ABM16a} proved that when the items are goods, the best a sequential picking algorithm can guarantee is an approximation of $\floor{(m-n+2)/2}$,
and such an approximation can be easily achieved by letting each of the first $n-1$ agents select one item and allocating all the remaining items to the last agent.
Compared to their result, we show that by carefully selecting the $a_{i}$'s, when items are chores,
we are able to significantly improve the bound to $O(\log {m\over n})$.

Next, we further improve the approximation ratio for cardinal and ordinal models by randomized algorithms.
Particularly, we show that by randomly allocating each item but allowing each agent to
recognize a set of ``bad'' items and to be able to decline each allocated ``bad'' item once,
the resulting algorithm is strategyproof and achieves an approximation ratio of $O(\sqrt{\log n})$ in expectation.

We complement these upper bound results with lower bound results by showing that for cardinal and ordinal model,
no deterministic strategyproof algorithm has a better than $\frac{4}{3}$ approximation.
For the ordinal model, we prove that the lower bound of $\frac{4}{3}$ cannot be improved by non-strategyproof algorithms.
For randomized non-strategyproof algorithms, this bound cannot be improved to be better than $\frac{6}{5}$.

Finally, for the public ranking model, we show that the greedy round-robin algorithm is strategyproof and 
achieves 2-approximation. 
This is also surprising as when the items are goods, the best known approximation is $O(n)$ by \citet{ABM16a}.
When $n\leq 3$, we give a strategyproof divide-and-choose algorithm that further improves this ratio to $\frac{3}{2}$.
We complement these upper bound results by a lower bound of $6\over 5$ for any deterministic strategyproof algorithms.


Our results as well as previous results for the case of goods are summarized in Table~\ref{table:lit}.



\subsection{Related Work}


MMS fairness is weaker than the proportionality fairness concept that requires that each agent gets at least $1/n$ of the total utility she has for the set of all items~\citep{BoLe15a}. However for indivisible items, there may not exist an allocation that guarantees any approximation for the proportionality requirement.

Most of the work on fair allocation of items is for the case of goods although recently,
fair allocation of chores~\citep{ARSW17a} or combinations of goods and chores~\citep{ACI19a} has received attention as well.
\citet{ARSW17a} proved that MMS allocations do not always exist but can be 2-approximated by a simple algorithm. \citet{ARSW17a} also presented a PTAS for relaxation of MMS called optimal MMS.
\citet{BM17a} presented an improved approximation algorithm for MMS allocation of chores. 

Strategyproofness is a challenging property to satisfy for fair division algorithms. \citet{ABM16a} initiated the work on strategyproof goods allocation with respect to MMS fairness. In particular they proved the results covered in the goods part of Table~\ref{table:lit}. There is also work on the approximation of welfare that can be achieved by strategyproof algorithms for allocation of \emph{divisible} items~\citep{ACF+16a}.

\section{Model and Preliminaries}

For the fair allocation problem, $N$ is a set of $n$ agents, and $M$ is a set of $m$ indivisible items.
The goal of the problem is to fairly distribute all the items to these agents.
Different agents may have different preferences for these items and these preferences are generally captured by
utility or {\em valuation} functions: each agent $i$ is associated with a function $v_i:2^{M}\to \mathbb{R}$ that valuates any set of items.

\paragraph{MMS fairness.}

Imagine that agent $i$ gets the opportunity to partition all the items into $n$ bundles, but she is the last to choose a bundle.
Then her best strategy is to partition the items such that the smallest value of a bundle is maximized.
Let $\Pi(M)$ denote the set of all possible $n$-partitionings of $M$.
Then the {\em maxmin share (MMS)} of agent $i$ is defined as
\begin{equation}\label{eq:mms:negative}
\MMS_i = \max_{\langle X_1, \ldots, X_n\rangle \in \Pi(M)} \min_{j \in N}  v_i(X_j).
\end{equation}
If agent $i$ finally receives a bundle of items with value at least $\MMS_i$,
she is happy with the final allocation.

\medskip
In this work, it is assumed that items are chores: $v_i(\{j\})\leq 0$ for all $i\in N$ and $j\in M$.
Then each agent actually wants to receive as few items as possible.
For ease of analysis, we ascribe a disutility or \emph{cost} function $c_i=-v_i$ for each agent $i$.
In this paper, we assume that the cost function of each agent $i$ is additive.
We represent each cost function $c_{i}$ by a vector $(c_{i1},\cdots,c_{im})$ where $c_{ij}=c_{i}(\{j\})$ is the cost of agent $i$ for item $j$.
Then for any $S\subseteq M$ we have $c_{i}(S)=\sum_{j\in S}c_{ij}$.
Agent $i$'s maxmin share can be equivalently defined as
\begin{equation}\label{eq:mms:positive}
\MMS_i = \min_{\langle X_1, \ldots, X_n\rangle \in \Pi(M)} \max_{j \in N}  c_i(X_j).
\end{equation}


Note that the maxmin threshold defined in Equation \ref{eq:mms:positive} is positive which is 
the opposite number of the threshold defined in Equation \ref{eq:mms:negative}.
Throughout the rest of our paper, we choose to use the second definition.
For each agent $i$, we use a permutation over $M$, $\sigma_{i}:M\to [m]$, to denote agent $i$'s {\em ranking} on the items: $c_{i\sigma_{i}(1)}\geq\cdots \geq c_{i\sigma_{i}(m)}$.
In other words, item $\sigma_i(1)$ is the least preferred item and $\sigma_i(m)$ is the most preferred.

Let $x=(x_{i})_{i\in N}$ be an {\em allocation}, where $x_{i}=(x_{ij})_{j\in M}$ and $x_{ij}\in \{0,1\}$ indicates if agent $i$ gets item $j$ under allocation $x$.
A feasible allocation guarantees a partition of $M$, i.e., $\sum_{i\in N}x_{ij}=1$ for any $j\in M$.
We somewhat abuse the definition and let $X_{i}=\{j\in M | x_{ij}=1\}$ and $c_{i}(x)=c_{i}(x_{i})=c_{i}(X_{i})$.
An allocation $x$ is called an {\em MMS allocation} if $c_{i}(x_{i}) \leq \MMS_{i}$ for every agent $i$
and {\em $\alpha$-MMS allocation} if $c_{i}(x_{i}) \leq \alpha\MMS_{i}$ for all agents $i$.

We first state the following simple properties of MMS.
Lemma~\ref{lem:mms:bound} implies if an agent receives $k$ items, then its cost is at most $k\cdot \MMS_i$.

\begin{lemma} \label{lem:mms:bound}
For any agent $i$ and any cost function $c_{i}$, 
\begin{itemize}
\item $\MMS_{i}\geq \frac{1}{n}c_{i}(M)$;
\item $\MMS_{i}\geq c_{ij}$ for any $j\in M$.
\end{itemize}
\end{lemma}
\begin{proof}
 The first inequality is clear as for any partition of the items, the largest bundle has cost at least $\frac{1}{n}c_{i}(M)$.

 For the second inequality, it suffices to show $\MMS_{i}\geq c_{i\sigma_{i}(1)}$.
 This is also clear since in any partitioning of the items, the largest bundle should have cost at least $c_{i \sigma_{i}(1)}$.
\end{proof}
 
By Lemma \ref{lem:mms:bound}, it is easy to see that if $m\leq n$, any allocation that allocates at most one item to each agent is an MMS allocation. Thus throughout this paper, we assume $m > n$.

\paragraph{Models.}
In the {\em cardinal model},
the agents are asked to express their cardinal costs over $M$.
A deterministic {\em cardinal algorithm} is denoted by a function 
$\mathcal{M}: (\mathbb{R}^{m})^{n} \to \Pi(M)$.
If an algorithm is restricted to only use the resulting rankings of the reported cardinal cost functions to allocate the items, we called it an {\em ordinal algorithm} and the corresponding problem is called the {\em ordinal model}.
If the algorithm has the information of all agents' rankings by default, and 
every agent has to report her cost function with respect to the known ranking,
the algorithm is called a {\em public ranking algorithm} and the corresponding problem is called the {\em public ranking model}.
A deterministic algorithm ${\cal M}$ is called ($\alpha$-approximate) MMS if
for any cost functions, it always outputs an ($\alpha$-) MMS allocation.
A randomized algorithm ${\cal M}$ returns a distribution over $\Pi(M)$ and is called $\alpha$-approximate MMS if for any cost functions $c_{1},\cdots, c_{n}$, 
$\mathbf{E}_{x\sim \mathcal{M}(c_{1},\cdots, c_{n})} [\max_{i\in N}\frac{c_{i}(x)}{\MMS_{i}}] \leq \alpha$.\footnote{Note that 
if the $\alpha$-approximation is defined as for every agent $i$, $\mathbf{E}_{x\sim \mathcal{M}(c_{1},\cdots, c_{n})}c_{i}(x)\leq \alpha\MMS_{i}$,
the problem becomes trivial as uniform-randomly allocating all items optimizes $\alpha$ to be 1.}

\medskip
In this work, we study the situation when the costs are private information of the agents. 
Each agent may withhold her true cost function in order to minimize her own cost for the allocation.
We call an algorithm \emph{strategyproof (SP)} if no agent can unilaterally misreport her cost function to reduce her cost.

Formally, a deterministic algorithm $\mathcal{M}$ is called SP if for every agent $i$, cost function $c_{i}$ and the cost functions $c_{-i}$ of other agents, $c_{i}(\mathcal{M}(c_{i},c_{-i}))\geq c_{i}(\mathcal{M}(c'_{i},c_{-i}))$ holds for all $c'_{i}$.
We call a randomized algorithm $\mathcal{M}$ {\em SP in expectation} if for every $i$, $c_{i}$ and $c_{-i}$, $\mathbf{E}_{x\sim \mathcal{M}(c_{i},c_{-i})}c_{i}(x)
\geq \mathbf{E}_{x\sim \mathcal{M}(c'_{i},c_{-i})} c_{i}(x)$ holds for all $c'_{i}$.

\begin{example}
	Suppose the cost function of an agent on four items is $c_{1} = (1,2,3,4)$.
	In an SP cardinal algorithm, reporting $c_{1}$ minimizes her cost (in expectation, for randomized algorithm and the same for the following cases);
	In an SP ordinal algorithm, reporting $c_{14}\geq c_{13} \geq c_{12} \geq c_{11}$ minimizes her cost;
	In an SP public ranking algorithm, the algorithm knows $c_{14}\geq c_{13} \geq c_{12} \geq c_{11}$ by default, and the agent minimizes her cost by reporting $c_{1}$. 
\end{example}


By the above definition, we have the following lemma immediately, which also appeared in \citet{ABM16a}.

\begin{lemma}\label{lem:model:imply}
	An SP $\alpha$-approximation algorithm for the ordinal model is also SP $\alpha$-approximate for the cardinal model.
	An SP $\alpha$-approximation algorithm for the cardinal model is also SP $\alpha$-approximate for the public ranking model.
%
\end{lemma}

We end this section by providing a necessary condition of all SP algorithms for cardinal and public ranking models, which is mainly used to prove our hardness results.

\begin{definition}
	An allocation algorithm $\mathcal{M}$ is {\em monotone} if for any cost functions $c_{1}, \cdots, c_{n}$ 
	and $x=\mathcal{M}(c_{1}, \cdots, c_{n})$,
	increasing $c_{ij}$ for some $x_{ij}=0$, or decreasing $c_{ij}$ for some $x_{ij} = 1$ does not change $x_i$.
\end{definition}

First, by perturbing the costs by arbitrarily small different values, we can assume without loss of generality that 
the cost $c_i(S)$ of agent $i$ is different for every $S\subsetneq M$.

\begin{lemma}\label{lemma:monotone}
	All SP algorithms are monotone.
\end{lemma}
\begin{proof}
	 Fix any agent $i$ and let $x$ be the allocation when $i$ reports $c_{i}$ and the others report $c_{-i}$.
 	We fist consider the case when $x_{ij}=0$ and $c_{ij}$ is increased.
 	Let $c'$ be the new cost profile, and $x'$ be the new allocation.
 	If $x'_{ij} = 1$, then if $c'_i(x'_{i}) \leq c_i(x_{i})$, then agent $i$ has incentive to lie when its true cost is $c$ (since $c_i(x'_i) < c'_i(x'_i)$); if $c'_i(x'_i) > c_i(x_i)$, then agent $i$ has incentive to lie when its true cost is $c'$.
 	Hence we have $x'_{ij} = 0$.
 	For the same reason, we should have $c_i(x'_i) = c_i(x_i)$, which implies $x'_i = x_i$.

 	Next, we consider the case when $x_{ij}=1$ and $c_{ij}$ is decreased.
 	If $x'_{ij} = 0$, then if $c'_i(x'_i) = c_i(x'_i) < c_i(x_i)$, then agent $i$ has incentive to lie when its true cost is $c$; if $c'_i(x'_i) = c_i(x'_i) \geq c_i(x_i)$, then agent $i$ has incentive to lie when its true cost is $c'$ (since $c'_i(x_i) < c_i(x_i)$).
 	Hence we have $x'_{ij} = 1$.
 	We further have $c_i(x'_i)-c_{ij} = c_i(x_i)-c_{ij}$
 	as otherwise agent $i$ has incentive to lie when its true cost is the one that results in a higher cost. Hence we have $x'_i = x_i$.
\end{proof}

\section{Ordinal Model: Deterministic Algorithms} \label{sec:ord:det}

Before we present our algorithm for the ordinal model,
we first discuss the limitation of deterministic ordinal algorithms.

\begin{lemma}\label{thm:4/3-hardness}
	No deterministic ordinal algorithm (even non-SP) has an approximation ratio smaller than $\frac{4}{3}$, even for $2$ agents and $4$ items.
\end{lemma}
\begin{proof}
	Consider the instance with $2$ agents, whose ranking on the $m=4$ items are identical. 
	Without loss of generality, assume the item with maximum cost is given to the first agent, i.e. $x_{11} = 1$.	
	If the first agent is allocated only one item, then for the case when $c_2 = (1,1,1,1)$, 
	the approximation ratio is $\frac{3}{2}$: the second agent has total cost $3$ while $\MMS_{2}=2$.
	Otherwise for the case when $c_1 = (3,1,1,1)$, the approximation ratio is at least $\frac{4}{3}$, as the first agent has total cost at least $3+1 = 4$ while $\MMS_{1}=3$.
\end{proof}


Next we present a deterministic sequential picking algorithm that is $O(\log\frac{m}{n})$-approximate and SP.
\citet{ABM16a} gave a deterministic SP ordinal algorithm which is $O(m-n)$-approximate when the items are goods.
In the following, we show that if all the items are chores, it is possible to improve the bound to $O(\log \frac{m}{n})$.
Without loss of generality, we assume that $n$ and $\frac{m}{n}$ are at least some sufficiently large constant.
As otherwise it is trivial to obtain an $O(1)$-approximation by assigning $\frac{m}{n}$ arbitrary items to each agents.

\begin{theorem}\label{th:strategyproof-ordinal}
	There exists a deterministic SP ordinal algorithm with approximation ratio $O(\log \frac{m}{n})$.
\end{theorem}

%

\paragraph{$\SP$.}
Fix a sequence of integers $a_1, \ldots, a_n$ such that $\sum_{i\leq n}a_i = m$.
Order the agents arbitrarily. For $i=n,n-1,\ldots,1$, let agent $i$ pick $a_i$ items from the remaining items.

\medskip

We note that as long as $a_i$'s do not depend on the valuations of agents, the rule discussed above is the serial dictatorship rule for multi-unit demands. 
When it is agent $i$'s turn to pick items, it is easy to see that her optimal strategy is to pick the top-$a_i$ items with smallest cost, among the remaining items.
Hence immediately we have the following lemma.

\begin{lemma}
	For any $\{a_i\}_{i\leq n}$, $\SP$ is SP. 
\end{lemma}

It remains to prove the approximation ratio.

\begin{lemma}\label{lem:ord:ai}
	There exists a sequence $\{a_i\}_{i\leq n}$ such that the approximation ratio of $\SP$ is $O(\log \frac{m}{n})$.
\end{lemma}
\begin{proof}
	We first establish a lower bound on the approximation ratio in terms of $\{a_i\}_{i\leq n}$. Then we show how to fix the numbers appropriately to get a small ratio.
	Let $r$ be the approximation ratio of the algorithm.
	
	Consider the moment when agent $i$ needs to pick $a_i$ items.
	Recall that at this moment, there are $\sum_{j\leq i} a_j$ items, and the $a_i$ ones with smallest cost will be chosen by agent $i$.
	Let $c$ be the average cost of items agent $i$ picks, i.e., $c_i(x) = c\cdot a_i$.
	On the other hand, each of the $\sum_{j\leq i-1} a_j$ items left has cost at least $c$.
	Thus we have $\MMS_i \geq c \cdot \left\lceil \frac{a_1+\ldots+a_{i-1}}{n} \right\rceil$ and
	\begin{equation*}
	r = \max_{i\in N}\left \{\frac{c_i(X_i)}{\MMS_i} \right\} \leq \max_{i\in N} \left\{\frac{a_i}{\left\lceil \frac{a_1+\ldots+a_{i-1}}{n} \right\rceil} \right\}.
	\end{equation*}
	
	It suffices to compute a sequence of $a_1,\ldots, a_n$ that sum to $m$ and minimizes this ratio.
	Fix $K= 2\log \frac{m}{n}$. Let
	\begin{equation*}
	a_i = \begin{cases}
	2, & i\leq \frac{n}{2}, \\
	\min\{ \textstyle m-\sum_{j<i}a_j, \left\lceil K\cdot (1+\frac{K}{n})^{i-\frac{n}{2}-1} \right\rceil \}, & i> \frac{n}{2}.
	\end{cases}
	\end{equation*}
	
	Note that the first term of the $\min$ is to guarantee we leave enough items for the remaining agents.
	Moreover, truncating $a_i$ is only helpful for minimizing the approximation ratio and thus we only need to consider the case when $a_i$ equals the second term of the $\min$.
	In the following, we show that
	\begin{enumerate}
		\item all items are picked: $\sum_{i\in N}a_i = m$;
		\item for every $i>\frac{n}{2}$: $a_i \leq K\cdot \left\lceil \frac{a_1+\ldots+a_{i-1}}{n} \right\rceil$.
	\end{enumerate}
	
	Note that for $i\leq \frac{n}{2}$, since agent $i$ receives $2$ items, the approximation ratio is trivially guaranteed.
	
	The first statement holds because
	\begin{align*}
	& \textstyle \sum_{i=1}^\frac{n}{2} 2 + \sum_{i=\frac{n}{2}+1}^n \left( K\cdot (1+\frac{K}{n})^{i-\frac{n}{2}-1} \right)\\
	= & \textstyle \sum_{i \leq \frac{n}{2}} \left( K\cdot (1+\frac{K}{n})^{i-1} \right) + n \\
	= & \textstyle (1+\frac{K}{n})^\frac{n}{2}\cdot n - n + n \geq 2^\frac{K}{2}\cdot n^2 > m,
	\end{align*}
	and $a_i$'s will be truncated when their sum exceeds $m$.
	
	For $i>\frac{n}{2}$, observe that (let $l = i-\frac{n}{2}-1$)
	\begin{align*}
	\textstyle \frac{1}{n}(a_1+\ldots+a_{i-1}) & \textstyle = 1 + \frac{1}{n}\sum_{j = 1}^{l} K\cdot (1+\frac{K}{n})^{j-1} \\
	& \textstyle = 1 + (1+\frac{K}{n})^{l} - 1 = (1+\frac{K}{n})^{l}.
	\end{align*}
	
	Thus we have $a_i \leq \left\lceil K\cdot (1+\frac{K}{n})^{l} \right\rceil \leq K\cdot \left\lceil  (1+\frac{K}{n})^{l} \right\rceil \leq K\cdot \left\lceil \frac{a_1+\ldots+a_{i-1}}{n} \right\rceil$, as claimed.
\end{proof}

We conclude the section by showing that our approximation ratio is asymptotically tight for $\SP$.

\begin{lemma}[Limits of $\SP$]
	The $\SP$ algorithm (with any $\{a_i\}_{i\in N}$) has approximation ratio $\Omega(\log \frac{m}{n})$.
\end{lemma}
\begin{proof}
	Fix $K = \frac{1}{4}\log \frac{m}{n}$. Suppose there exists a sequence of $\{a_i\}_{i\in N}$ such that the algorithm is $K$-approximate.
	
	Then the last agent to act must receive at most $K$ items, i.e., $a_1 \leq K$.
	Next we show by induction on $i=2,3,\ldots,n$ that $a_i \leq K(1+\frac{2K}{n})^{i-1}$ for all $i\in N$.
	
	Suppose the statement is true for $a_1,\ldots, a_i$. Then if $a_{i+1} > K(1+\frac{2K}{n})^{i}$, we have
	\begin{equation*}
	\frac{a_{i+1}}{a_1+\ldots+a_{i+1}} > \frac{ K(1+\frac{2K}{n})^{i}}{k\cdot\frac{n}{2K}((1+\frac{2K}{n})^{i+1}-1)} \geq \frac{K}{n}.
	\end{equation*}
	
	Thus we have $\sum_{i=1}^n a_i \leq n\cdot\left( (1+\frac{2K}{n})^n - 1 \right)
	\leq n\cdot \left( e^{2K} - 1 \right) < m$,	which is a contradiction, since not all items are allocated.
\end{proof}

%

\section{Ordinal Model: Randomized Algorithms} \label{sec:ord:ran}

We have shown a logarithmic approximation algorithm $\SP$ for the problem.
However, the algorithm may still have poor performance when the number of items is much larger than the number of agents, e.g., $m = 2^n$.
In this section we present a randomized $O(\sqrt{\log n})$-approximation ordinal algorithm, 
which is SP in expectation. 

Again, before we show our algorithm, let us first see a limitation of the randomized ordinal algorithms.

\begin{lemma}\label{thm:ordi:rand:hardness}
	No randomized ordinal algorithm (even non-SP) has approximation ratio smaller than $\frac{6}{5}$, even for $2$ agents and $4$ items.
\end{lemma}
\begin{proof}
	Consider the instance with $2$ agents, whose ranking on the $m=4$ items are identical. 
	Let $p$ be the probability that the algorithm assigns $2$ items to both agents.
	
	If $p\leq \frac{3}{5}$, consider the instance with evaluation $(1,1,1,1)$, for which $\MMS_1 = \MMS_2 = 2$.
	Then with probability $1-p$, the agent receiving at least $3$ items has cost at least $\frac{3}{2}$ times its maximin share, which implies that the expected approximation ratio is at least $p+(1-p)\cdot\frac{3}{2} = \frac{3}{2} - \frac{p}{2} \geq \frac{6}{5}$.
	
	If $p> \frac{3}{5}$, then consider the instance with evaluation $(3,1,1,1)$, for which $\MMS_1 = \MMS_2 = 3$.
	Then with probability $p$, the agent receiving the item with cost $3$ has cost at least $\frac{4}{3}$ times its maximin share, which implies an expected approximation ratio at least $p\cdot \frac{4}{3}+(1-p) = 1 + \frac{p}{3} \geq \frac{6}{5}$.
\end{proof}


Basically, if we randomly allocate all the items, one is able to show that the algorithm achieves an approximation of $O(\log n)$.
The drawback of this na{\" i}ve randomized algorithm is that it totally ignores the rankings of agents.
In the following, we show that if the agents have opportunities to decline some ``bad'' items,
the performance of this randomized algorithm improves to $O(\sqrt{\log n})$.
Note that since we already have an $O(\log \frac{m}{n})$-approximate deterministic algorithm for the ordinal model, 
it suffices to consider the case when $m \geq n\log n$.

\paragraph{$\RD$.}
Let $K=\lfloor n\sqrt{\log n} \rfloor$.
Based on the ordering of items submitted by agents, for each agent $i$, 
label the $K$ items with largest cost as ``large'', and the remaining to be ``small''. 
It can be also regarded as each agent reports a set $M_{i}$ of large items with $|M_{i}|=K$.
The algorithm operates in two phases.
\begin{itemize}
	\item Phase 1: every item is allocated to a uniformly-at-random chosen agent, independently.
	After all allocations, gather all the large items assigned to every agent into set $M_b$. Note that $M_{b}$ is also a random set.
	
	\item Phase 2: Redistribute the items in $M_{b}$ evenly to all agents: every agent gets $\frac{|M_b|}{n}$ random items.
\end{itemize}

\begin{theorem}\label{th:strategyproof-ordinal:random}
	There exists a randomized SP ordinal algorithm with approximation ratio $O(\sqrt{\log n})$.
\end{theorem}

We prove Theorem~\ref{th:strategyproof-ordinal:random} in the following two lemmas.

\begin{lemma}
	In expectation, the approximation ratio of Algorithm $\RD$ is $O(\sqrt{\log n})$.
\end{lemma}
\begin{proof}
	We show that with probability at least $1-\frac{2}{n}$, every agent $i$ receives a collection of items of cost at most $O(\sqrt{\log n})\cdot \MMS_i$.
	Fix any agent $i$. 
	Without loss of generality, we order the items according to agent $i$'s ranking, 
	i.e., $\sigma_{i}(j)=j$ for any $j\in M$ and $c_{i1}\geq \cdots \geq c_{im}$.

	For ease of analysis, we rescale the costs such that
	\begin{equation*}
	c_{i1}+c_{i2}+\ldots+c_{im} = n\sqrt{\log n} = K.
	\end{equation*}
	
	Note that after the scaling, agent $i$'s maximin share is $\MMS_{i}\geq \sqrt{\log n}$.
	Let $x_{ij}$ denote the random variable indicating that the contribution of item $j$ to the cost of agent $i$.
	Then for $j > K$, $x_{ij} = c_{ij}$ with probability $\frac{1}{n}$, and $x_{ij} = 0$ otherwise.
	For $j \leq K$, $x_{ij} = 0$ with probability $1$.
	Note that
	\begin{equation*}
	\textstyle \mathbf{E}[\sum_{i=1}^m x_i ] = \frac{1}{n}\cdot \sum_{i=K+1}^m c_{ij} \leq \frac{K}{n} = \sqrt{\log n}.
	\end{equation*}
	
	Moreover, we have $c_{ij} \leq 1$ for $j > K$, as otherwise we have the contradiction that $\sum_{j=1}^{K}c_{ij} > K$.
	Note that $\{x_{ij}\}_{j\leq m}$ are independent random variables taking value in $[0,1]$.
	Hence by Chernoff bound we have
	\begin{align*}
	& \textstyle \Pr[ \sum_{j=1}^m x_{ij} \geq 7\sqrt{\log n} \cdot \MMS_{i} ] \\
	\leq & \textstyle  \Pr[ \sum_{j=1}^m x_{ij} \geq 7\log n ] \\
	\leq &\textstyle  \exp\left(-\frac{1}{3}\cdot \left(\frac{7\log n}{\mathbf{E}[\sum_{i=1}^m x_i]}-1\right) \cdot \mathbf{E}[\sum_{i=1}^m x_i] \right) < \frac{1}{n^2}.
	\end{align*}
	
	Then by union bound over the $n$ agents, we conclude that with probability at least $1-\frac{1}{n}$, every agent $i$ receives a bundle of items of cost at most $O(\sqrt{\log n})\cdot \MMS_i$ in phase 1.
	
	Now we consider the items received by an agent in the second phase.
	Recall that the items $M_b$ will be reallocated evenly.
	By the second argument of Lemma \ref{lem:mms:bound}, to show that every agent $i$ receives a bundle of items of cost $O(\sqrt{\log n})\cdot \MMS_i$ in the second phase, it suffices to prove that $|M_b| = O(n\sqrt{\log n})$ (with probability at least $1-\frac{1}{n}$).
	
	Let $y_j\in\{0,1\}$ be the random variable indicating whether item $j$ is contained in $M_b$.
	For every item $j$, let $b_j=|\{k : j \in M_{k}\}|$ be the number of agents that label item $j$ as ``large''.
	Then we have $y_j = 1$ with probability $\frac{b_j}{n}$.
	Since every agent labels exactly $n\sqrt{\log n}$ items, we have
	\begin{equation*}
	\textstyle \mathbf{E}[|M_b|] = \mathbf{E}[\sum_{i=1}^m y_i] = \frac{1}{n}\sum_{i=1}^m b_i = n\sqrt{\log n}.
	\end{equation*}
	
	Applying Chernoff bound we have
	\begin{align*}
	\textstyle \Pr[ \sum_{i=1}^m y_i \geq 2n\sqrt{\log n}] \leq \exp\left(-\frac{n\sqrt{\log n}}{3} \right) < \frac{1}{n}.
	\end{align*}
	
	Thus, with probability at least $1-\frac{2}{n}$, every agent $i$ receives a bundle of items with cost $O(\sqrt{\log n}\cdot \MMS_i)$ in the two phases combined.
	Since in the worse case, $i$ receives a total cost of at most $n\cdot \MMS_i$, in expectation, the approximation ratio is $(1-\frac{2}{n})\cdot O(\sqrt{\log n}) + \frac{2}{n}\cdot n = O(\sqrt{\log n})$.
\end{proof}


\begin{lemma}
	$\RD$ is SP in expectation.
\end{lemma}
\begin{proof}
	To prove that the algorithm is SP in expectation, it suffices to show that for every agent, the expected cost it is assigned is minimized when being truthful.
	Let $K=n\sqrt{\log n}$ and fix any agent $i$.
	Suppose $c_{i1},\ldots,c_{iK}$ are the costs of items labelled ``large'' by the agent; and $c_{i,K+1},\ldots,c_{im}$ are the remaining items.
	Then the expected cost assigned to the agent in the first phase is given by $\frac{1}{n}\sum_{j=K+1}^m c_{ij}$, as every item is assigned to the agent with probability $\frac{1}{n}$.
	Now we consider the cost the agent is assigned in the second phase.
	
	Recall that the expected total cost of items to be reallocated in the second phase is $\mathbf{E}[\sum_{j\in M_b} c_{ij}] = \sum_{j=1}^m c_{ij}\cdot \frac{b_j}{n}$, where $b_j$ is the number of agents that label item $j$ ``large''.
	Let $\mathcal{E}$ be this expectation when agent $i$ does not label any item ``large''. 
		
	By labelling $c_{i1},\ldots, c_{iK}$ ``large'', agent $i$ increases the probability of each item $j\leq K$ being included in $M_b$ by $\frac{1}{n}$.
	Thus it contributes an $\frac{1}{n}\sum_{j=1}^K c_{ij}$ increase to the expectation of total cost of $M_b$.
	In other words, we have
	\begin{equation*}
	\textstyle \mathbf{E}[\sum_{j\in M_b} c_{ij}] = \mathcal{E} + \frac{1}{n}\sum_{j=1}^K c_{ij}.
	\end{equation*}
	
	Since a random subset of $\frac{|M_b|}{n}$ items from $M_b$ will be assigned to agent $i$, the expected total value of items assigned to the agent in the two phases is given by
	\begin{equation*}
	\textstyle \frac{1}{n}\sum_{j=K+1}^m c_{ij} + \frac{1}{n}\cdot \left( \mathcal{E} + \frac{1}{n}\sum_{j=1}^K c_{ij} \right).
	\end{equation*}
	
	Obviously, the expression is minimized when $c_{i1}+\ldots+c_{iK}$ is maximized.
	Hence every agent minimizes its expected cost by telling the true ranking over the items. 
\end{proof}

\section{Cardinal Model}

First, we present a lower bound on the approximation ratio for the all deterministic SP cardinal algorithms.

\begin{lemma}\label{thm:3/2-hardness}
	No deterministic cardinal SP algorithm has approximation ratio smaller than $\frac{4}{3}$, even for $2$ agents and $4$ items.
\end{lemma}
\begin{proof}
	First, consider $c_{1}=c_{2}=(3,1,1,1)$, for which $\MMS_{1}=\MMS_{2}=3$.
	To obtain an approximation smaller than $\frac{4}{3}$, the only possible allocation is to assign the first item to some agent, and the remaining items to the other.
	Without loss of generality, suppose agent $1$ receives the first item.
	
	By monotonicity of SP algorithms (Lemma~\ref{lemma:monotone}), for the case when $c_{1}=(1,3,1,1)$ and $c_{2}=(3,1,1,1)$. The assignment remains unchanged.
	Now we consider the profile when $c_{1}=(1,3,1,1)$ and $c_{2}=(2,2,1,1)$.
	
	Note that we also have $\MMS_{1}=\MMS_{2}=3$, and thus (to guarantee the approximation ratio) agent $2$ cannot receive the last three items.
	Moreover, to guarantee SPness, agent $2$ cannot receive a proper subset of the last three items, as otherwise agent $2$ will misreport $(2,2,1,1)$ when its true value is $(3,1,1,1)$.
	Thus the first item must be assigned to agent $2$, and consequently the second item must be assigned to agent $1$.
	To guarantee the approximation ratio, agent $1$ should not receive any other item, which means that agent $2$ must receive the last two items, which violates the better-than-$\frac{4}{3}$ approximation ratio.
\end{proof}

For positive results, by Lemma \ref{lem:model:imply}, both algorithms in Sections \ref{sec:ord:det} and \ref{sec:ord:ran} apply to the cardinal model.
Thus we immediately have the following.

\begin{corollary}\label{th:strategyproof-ordinal-cardinal}
	For the cardinal model, there exists a deterministic SP algorithm with approximation ratio $O(\log \frac{m}{n})$;
	and a randomized SP-in-expectation algorithm with approximation ratio $O(\sqrt{\log n})$.
\end{corollary}


\section{Public Ranking Model}

\citet{ABM16a} provided a deterministic SP public ranking algorithm which is $O(n)$-approximate if the items are goods.
In this section, we show that if the items are chores, we can do much better.

\subsection{Strategyproof Algorithms}

We first give a simple SP algorithm with an approximation ratio of at most 2.

\paragraph{$\RR$.}
Fix an arbitrary order of the agents, let the agents pick items in the round-robin manner. 

\begin{theorem}\label{th:2upper}
	For the public ranking model, $\RR$ is SP and has an approximation ratio of $2-\frac{1}{n}$.
\end{theorem}
\begin{proof}
	\citet{ARSW17a} proved that $\RR$ gives an approximation bound of $2-\frac{1}{n}$ if at each round every agent is allocated the item with smallest cost.
	Note that the algorithm is ordinal hence no agent can change the outcome by misreporting her cardinal utilities. Furthermore, since the preference rankings of the agents are public knowledge, agents cannot misreport by expressing a different ordinal preference. Hence the algorithm is SP for the public ranking model.
\end{proof}

By Theorem \ref{th:2upper}, when $n=2$ and $3$, the algorithm gives a $\frac{3}{2}$ and $\frac{5}{3}$ approximations, respectively.
Indeed, for $n=3$, we show a divide-and-choose algorithm which is SP and still guarantees $\frac{3}{2}$-approximation ratio.

\begin{theorem}\label{th:pr:n=3}
	For the public ranking model, there exist an SP $1.5$-approximation algorithm when $n=3$.
\end{theorem}
\begin{proof}
	Without loss of generality, we order the items according to agent 1's ranking, 
	i.e., $\sigma_{1}(j)=j$ for any $j\in M$ and $c_{11}\geq \cdots \geq c_{1m}$.
	The algorithm runs as follows:
	\begin{itemize}
		\item Let $S_{1}=\{1\}$, $S_{2}=\{j\in M| j \bmod 2=0\}$ and $S_{3}=\{j\in M| j>1 \mbox{ and } j \bmod 2=1\}$. Note that $(S_{1},S_{2},S_{3})$ is a partition of all the items.
		\item Let agent 2 select her favourite bundle from $S_{1}$, $S_{2}$, $S_{3}$.
		\item Let agent 3 select her favourite one from the two bundles left in Step 2, and assign the last bundle to agent 1.
	\end{itemize}
	It is easy to see that the above algorithm is SP as the algorithm does not use any information reported by agent 1
	and both of agent 2 and agent 3's best strategy is to report the costs such that the bundle with smallest cost is selected.
	We are left to prove the approximation ratio. 
	
	For agent 1, note that $c_{1}(S_{1}\cup S_{3}) \geq c_{1}(S_{2}) \geq c_{1}(S_{3})$.
	By the first argument of Lemma~\ref{lem:mms:bound}, $c_{1}(S_{3}) \leq c_{1}(S_{2}) \leq \frac{1}{2}c_{1}(M) \leq \frac{3}{2}\cdot\MMS_{1}$.
	By the second argument of Lemma~\ref{lem:mms:bound}, $c_{1}(S_{1}) \leq \MMS_{1}$.
	That is, no matter which bundle is left, the cost of this bundle is at most $\frac{3}{2}\cdot\MMS_{1}$ to agent 1.
	
	As agent 2 gets her best bundle, her cost is at most $\MMS_2$.
	
	For agent 3, since she is still able to select one from two bundles,
	her cost for the better bundle is at most $\frac{1}{2}\cdot c_{3}(M)$ which is at most $\frac{3}{2}\cdot \MMS_{3}$ by Lemma~\ref{lem:mms:bound}.
	
	In conclusion, the algorithm is a $\frac{3}{2}$-approximation.
\end{proof}

\subsection{Limitation of Strategyproof Algorithms}
Next, we complement the upper bound results with a lower bound result for the the public ranking model.

\begin{lemma}\label{th:6/5-hardness}
	For the public ranking model, no deterministic SP algorithm has an approximation ratio smaller than $\frac{6}{5}$, even for $2$ agents.
\end{lemma}
\begin{proof}
	Let $n=2$ and $m=6$.
	Assume for contradiction that there exists some SP algorithm with approximation ratio less than $\frac{6}{5}$.
	Suppose $c_1=(1,1,1,1,1,1)$ and $c_2=(1,1,1,1,1,1)$, then we have $\MMS_1 = \MMS_2 = 3$.
	The algorithm must assign every agent exactly $3$ items, as otherwise the cost for one of them is $4 \geq \frac{6}{5}\cdot 3$.
	Without loss of generality, assume that $A_1 = \{1,2,3\}$ and $A_2 = \{4,5,6\}$.
	
	By Lemma~\ref{lemma:monotone}, when $c_1=(1,1,1,1,1,5)$ and $c_2=(1,1,1,1,1,1)$, the assignments remain unchanged, i.e., $A_1 = \{1,2,3\}$.
	If we further increase $c_{26}$ to $5$, we have $\MMS_1 = \MMS_2 = 5$.
	If $\{6\}\subseteq A_2$, we should have $A_2 = \{6\}$ to guarantee an approximation ratio less than $\frac{6}{5}$.
	Then we know that agent $2$ has incentive to lie when its true cost is $c_2=(1,1,1,1,1,1)$.
	Hence we have $\{6\}\subseteq A_1$, which implies $A_1 = \{6\}$ and $A_2 = \{1,2,3,4,5\}$.
	Again, by Lemma~\ref{lemma:monotone}, for $c_1=(1,1,1,1,1,3)$ and $c_2=(1,1,1,1,1,5)$, the assignment remains unchanged.
	
	Applying a similar argument, for $c_1=(1,1,1,1,1,3)$ and $c_2=(1,1,1,1,1,1)$, we have $A_1 = \{1,2,3\}$.
	Then for $c_1=(1,1,1,1,1,3)$ and $c_2=(1,1,1,1,1,3)$, we have $\{6\}\subseteq A_1$, $|A_1|\leq 2$, and $c_2(A_2) = 4$ to guarantee an approximation ratio less than $\frac{6}{5}$.
	Thus by Lemma~\ref{lemma:monotone}, for $c_1=(1,1,1,1,1,3)$ and $c_2=(1,1,1,1,1,5)$, the assignment does not change, i.e., $|A_2| = 4$, which contradicts the conclusion we draw in the previous paragraph.
\end{proof}

\section{Conclusion}

In this paper, we initiated the study of SP and approximately maxmin fair algorithms for chore allocation. 
Our study leads to several new questions. 
The most obvious research questions would be to close the gap between the lower and upper approximation bounds for SP algorithms 
and to study the lower bound of randomized SP cardinal algorithms.

At present we have two parallel lines of research for goods and chores. 
It is interesting to consider similar questions for combinations of goods and chores~\citep{ACI19a}.
Another direction is to study the SP fair allocation algorithms for the case of asymmetric agents~\citep{Aziz19a}.

\newpage
\bibliographystyle{named}
\bibliography{abb,fair}

\begin{thebibliography}{}

\bibitem[\protect\citeauthoryear{Amanatidis \bgroup \em et al.\egroup
  }{2015}]{AMNS15a}
G.~Amanatidis, E.~Markakis, A.~Nikzad, and A.~Saberi.
\newblock Approximation algorithms for computing maximin share allocations.
\newblock In {\em Proceedings of the 35th International Colloquium on Automata,
  Languages, and Programming (ICALP)}, pages 39--51, 2015.

\bibitem[\protect\citeauthoryear{Amanatidis \bgroup \em et al.\egroup
  }{2016}]{ABM16a}
G.~Amanatidis, G.~Birmpas, and E.~Markakis.
\newblock On truthful mechanisms for maximin share allocations.
\newblock In {\em Proceedings of the Twenty-Fifth International Joint
  Conference on Artificial Intelligence, {IJCAI} 2016, New York, NY, USA, 9-15
  July 2016}, pages 31--37, 2016.

\bibitem[\protect\citeauthoryear{Aziz \bgroup \em et al.\egroup
  }{2016}]{ACF+16a}
H.~Aziz, J.~Chen, A.~Filos-Ratsikas, S.~Mackenzie, and N.~Mattei.
\newblock Egalitarianism of random assignment mechanisms.
\newblock In {\em Proceedings of the 15th International Conference on
  Autonomous Agents and Multiagent Systems (AAMAS)}, 2016.

\bibitem[\protect\citeauthoryear{Aziz \bgroup \em et al.\egroup
  }{2017}]{ARSW17a}
H.~Aziz, G.~Rauchecker, G.~Schryen, and T.~Walsh.
\newblock Algorithms for max-min share fair allocation of indivisible chores.
\newblock In {\em Proceedings of the 31st AAAI Conference on Artificial
  Intelligence (AAAI)}, pages 335--341, 2017.

\bibitem[\protect\citeauthoryear{Aziz \bgroup \em et al.\egroup
  }{2018}]{ACI19a}
H.~Aziz, I.~Caragiannis, and A.~Igarashi.
\newblock Fair allocation of combinations of indivisible goods and chores.
\newblock {\em CoRR}, abs/1807.10684, 2018.

\bibitem[\protect\citeauthoryear{Aziz \bgroup \em et al.\egroup
  }{2019}]{Aziz19a}
H.~Aziz, H.~Chan, and B.~Li.
\newblock Weighted maxmin fair share allocation of indivisible chores.
\newblock In {\em Proceedings of the 28th International Joint Conference on
  Artificial Intelligence (IJCAI)}, 2019.

\bibitem[\protect\citeauthoryear{Barman and Murthy}{2017}]{BM17a}
S.~Barman and S.~Kumar~Krishna Murthy.
\newblock Approximation algorithms for maximin fair division.
\newblock In {\em Proceedings of the 18th ACM Conference on Economics and
  Computation (ACM-EC)}, pages 647--664, 2017.

\bibitem[\protect\citeauthoryear{Bouveret and Lema{\^\i}tre}{2016}]{BoLe15a}
S.~Bouveret and M.~Lema{\^\i}tre.
\newblock Characterizing conflicts in fair division of indivisible goods using
  a scale of criteria.
\newblock {\em Autonomous Agents and Multi-Agent Systems}, 30(2):259--290,
  2016.

\bibitem[\protect\citeauthoryear{Bouveret \bgroup \em et al.\egroup
  }{2016}]{BCM15a}
S.~Bouveret, Y.~Chevaleyre, and N.~Maudet.
\newblock Fair allocation of indivisible goods.
\newblock In F.~Brandt, V.~Conitzer, U.~Endriss, J.~Lang, and A.~D. Procaccia,
  editors, {\em Handbook of Computational Social Choice}, chapter~12. Cambridge
  University Press, 2016.

\bibitem[\protect\citeauthoryear{Budish}{2011}]{Budi11a}
E.~Budish.
\newblock The combinatorial assignment problem: Approximate competitive
  equilibrium from equal incomes.
\newblock {\em Journal of Political Economy}, 119(6):1061--1103, 2011.

\bibitem[\protect\citeauthoryear{Ghodsi \bgroup \em et al.\egroup
  }{2018}]{SGHSY18}
M.~Ghodsi, M.~HajiAghayi, M.~Seddighin, S.~Seddighin, and H.~Yami.
\newblock Fair allocation of indivisible goods: Improvements and
  generalizations.
\newblock In {\em Proceedings of the 19th ACM Conference on Economics and
  Computation (ACM-EC)}. ACM Press, 2018.

\bibitem[\protect\citeauthoryear{Procaccia and Tennenholtz}{2013}]{PrTe13a}
A.~D. Procaccia and M.~Tennenholtz.
\newblock Approximate mechanism design without money.
\newblock {\em ACM Transactions on Economics and Computation}, 1(4), 2013.

\bibitem[\protect\citeauthoryear{Procaccia and Wang}{2014}]{PrWa14a}
A.~D. Procaccia and J.~Wang.
\newblock Fair enough: Guaranteeing approximate maximin shares.
\newblock In {\em Proceedings of the 15th ACM Conference on Economics and
  Computation (ACM-EC)}, pages 675--692. ACM Press, 2014.

\end{thebibliography}

\end{document}